\title{Private Information Retrieval in\\Graph Based Replication Systems}
\documentclass[11pt,onecolumn]{IEEEtran}
\usepackage[boxruled,linesnumbered]{algorithm2e}

\usepackage{url}

\usepackage{multirow}
\usepackage[table,xcdraw]{xcolor}

\usepackage[numbers,sort]{natbib}

\usepackage{pgf,tikz}
\usepackage{mathrsfs}
\usetikzlibrary{arrows,patterns}
\usetikzlibrary{topaths,calc}

\usepackage{bm}

\usepackage{amssymb}

\usepackage{amsmath}

\usepackage{amsthm}

\usepackage{mathtools}

\newtheorem{theorem}{Theorem}

\newtheorem{lemma}[theorem]{Lemma}
\newtheorem{remark}[theorem]{Remark}

\newtheorem{corollary}[theorem]{Corollary}
\newtheorem{example}[theorem]{Example}
\newtheorem{proposition}[theorem]{Proposition}

\usepackage{hyperref}
\hypersetup{
	colorlinks,
	linkcolor={blue!100!black},
	citecolor={blue!100!black},
	urlcolor={blue!80!black}
}

\newcommand{\bF}{\mathbb{F}}

\newcommand{\cA}{\mathcal{A}}
\newcommand{\cB}{\mathcal{B}}
\newcommand{\cC}{\mathcal{C}}
\newcommand{\cD}{\mathcal{D}}

\newcommand{\cK}{\mathcal{K}}
\newcommand{\cL}{\mathcal{L}}
\newcommand{\cM}{\mathcal{M}}
\newcommand{\cN}{\mathcal{N}}

\newcommand{\cR}{\mathcal{R}}
\newcommand{\cS}{\mathcal{S}}
\newcommand{\cT}{\mathcal{T}}

\newcommand{\cV}{\mathcal{V}}

\newcommand{\bolda}{\textbf{a}}
\newcommand{\boldb}{\textbf{b}}
\newcommand{\boldc}{\textbf{c}}

\newcommand{\bolde}{\textbf{e}}

\newcommand{\boldq}{\textbf{q}}

\newcommand{\boldv}{\textbf{v}}

\newcommand{\boldx}{\textbf{x}}
\newcommand{\boldy}{\textbf{y}}

\DeclareMathOperator{\diag}{diag}
\DeclareMathOperator{\rank}{rank}
\DeclareMathOperator{\colspan}{colspan}
\DeclareMathOperator{\dist}{dist}

\newcommand{\boldmu}{\boldsymbol{\mu}}
\newcommand{\boldeta}{\boldsymbol{\eta}}
\newcommand{\boldgamma}{\boldsymbol{\gamma}}

\newcommand{\boldalpha}{\boldsymbol{\alpha}}

\DeclareSymbolFont{bbold}{U}{bbold}{m}{n}
\DeclareSymbolFontAlphabet{\mathbbold}{bbold}
\newcommand{\1}{\mathbbold{1}}

\author{\textbf{Netanel Raviv}$^\star$, \textbf{Itzhak Tamo}$^\dagger$, and \textbf{Eitan Yaakobi}$^\ddagger$\\
	\IEEEauthorblockA{\normalsize {$^\star$Department of Electrical Engineering, California Institute of Technology, Pasadena, CA 91125, USA.\\
			$^\dagger$Department of Electrical Engineering--Systems, Tel-Aviv University, Tel-Aviv 39040, Israel.\\
			$^\ddagger$Department of Computer Science, Technion---Israel Institute of Technology, Haifa 3200003, Israel}
		\thanks{Parts of this work were presented at the International Symposium on Information Theory (ISIT), Vail, CO, USA, 2018.}}}

\begin{document}
\maketitle
\thispagestyle{empty}


\begin{abstract}
In a Private Information Retrieval (PIR) protocol, a user can download a file from a database without revealing the identity of the file to each individual server. A PIR protocol is called \textit{$t$-private} if the identity of the file remains concealed even if~$t$ of the servers collude. Graph based replication is a simple technique, which is prevalent in both theory and practice, for achieving erasure robustness in storage systems. In this technique each file is replicated on two or more storage servers, giving rise to a (hyper-)graph structure. In this paper we study private information retrieval protocols in graph based replication systems. The main interest of this work is maximizing the parameter~$t$, and in particular, understanding the structure of the colluding sets which emerge in a given graph. Our main contribution is a $2$-replication scheme which guarantees perfect privacy from acyclic sets in the graph, and guarantees partial-privacy in the presence of cycles. Furthermore, by providing an upper bound, it is shown that the PIR rate of this scheme is at most a factor of two from its optimal value for an important family of graphs. Lastly, we extend our results to larger replication factors and to graph-based coding, which is a similar technique with smaller storage overhead and larger PIR rate.

\end{abstract}

\section{Introduction}
Recent data breaches in major corporations have emphasized the need for privacy in the digital era. Among the many challenges that designers of distributed storage systems face is the ability to support \textit{private information retrieval} (PIR) protocols. These protocols enable the end user to retrieve an entry of the database, while concealing the identity of that entry from the servers. This paper studies PIR protocols in a particular common type of distributed storage systems.

Coding for storage systems has developed tremendously in recent years. However, many system designers still favor replication techniques, over more involved ones, as a means to guarantee robustness against hardware failures~\cite{Hadoop,Cassandra}. In spite of having high storage overhead and low failure resilience, replication is often preferred due to its simplicity of implementation. In addition, various types of replication systems are studied in theoretical research due to their real-world impact and ease of analysis~\cite{FR,Lev1,Lev2,SalimFR,ExpanderCodes}. However, since contemporary datasets are far too large to be stored on one machine, it is usually the case where every machine stores a small number of selected files from the dataset, each of which is replicated among geographically separated machines. In turn, such systems can be modeled as hypergraphs, where nodes represent storage servers and (hyper-)edges represent files. In these graphs, an edge is incident with a node if a copy of the respective file is stored on the respective server. Storage systems which broadly adhere to the above outline are called \textit{graph-based replication systems}. A graph based replication system in which every file is replicated~$r$ times is called an~\textit{$r$-replication system}, and~$r$ is called its \textit{replication factor}.

One of the most important metrics by which PIR protocols are measured is their \textit{collusion resistance}. In its most simplistic form, a PIR protocol must guarantee perfect privacy against every individual server\footnote{In some settings, only computational privacy is required, but this paper focus exclusively on perfect privacy.}. That is, it should be computationally impossible for every individual server to infer any information regarding the identity of the requested file. The term collusion resistance measures the ability of a PIR protocol to perform beyond this baseline. That is, what is the maximum number of servers that still remain completely oblivious to the identity of the file, even if collusion among them is permitted. Traditionally, the term ``collusion'' stems from a mindset which considers the servers themselves as adversaries. Yet, the authors of this paper deem this interpretation obsolete, since it does not align with contemporary storage services. Instead, one can think of geographically separated servers as having independent security protocols, that must be individually broken by an adversary. In this case, the term ``colluding servers'' refers to a set of servers whose security was breached by an outside adversary, that can therefore observe their input and output. Normally, the term \textit{$t$-privacy} of a given protocol indicates the maximum number of servers that cannot infer any information regarding the identity of the file even if they collude; and in our alternative viewpoint, $t+1$ is the minimum number of individually-secured servers that must be breached by an adversary in order to infringe the perfect privacy of the protocol. Nevertheless, in our choice of terms we comply with the standard nomenclature.

PIR protocols have been studied extensively in the past years, and many additional metrics of interest were defined. Among the metrics of interests are: (a) the \textit{PIR rate}, which measures the ratio between the size of the desired data and the size of the downloaded one; (b) the \textit{upload complexity}, which measures the size of the queries that are sent to the servers; and (c) the \textit{storage overhead}, which measures the amount of redundancy in the system. While our main concern is understanding the collusion resistance of the system, we also address some of these metrics in our analysis.

In this paper we initiate a study about PIR protocols in graph based replication systems, and our primary focus is studying their collusion resistance. Since such systems are inherently non-uniform, in the sense that every server stores a different part of the dataset, one might expect that the collusion resistance will act accordingly. Indeed, our results show that the right viewpoint for analyzing colluding sets is not their size, but rather the structure of their induced subgraph. In particular, perfect privacy is maintained if the colluding sets do not contain certain sub-graphs. 

Our results shed light on the design of such systems in a bilateral manner. On one hand, we provide recommendations for system designers regarding the file dispersion in the system. On the other hand, we provide a way for analyzing the collusion resistance of a given system. In particular, we provide a PIR protocol for $2$-replication systems and show that its PIR rate at least half of its optimal value in many cases of interest.
For larger replication factors we provide a simple scheme whose collusion resistance is less than the replication factor, and another scheme which obtains a larger collusion resistance by a reduction to the two $2$-replication case. 


Further, we suggest an alternative \textit{graph-based coding} approach, in which every file is coded by using an MDS code, and the resulting codeword symbols are dispersed as in graph-based replication systems. While this approach reduces the storage overhead and increases the PIR rate, it requires a careful file dispersion in order to guarantee high collusion resistance. The results in this paper, and graph-based coding in particular, call for future research and practical implementations, that would hopefully bring the vast PIR literature closer to realistic storage systems.

This paper is structured as follows. Preliminaries and previous works are discussed in Section~\ref{section:preliminaries}. Protocols and bounds for $2$-replication systems are given in Section~\ref{section:Replication2}, and larger replication factors are discussed in Section~\ref{section:largerReplicationFactor}. Then, graph-based coding is discussed in Section~\ref{section:graphBasedCoding}, and open problems for future research are discussed in Section~\ref{section:Discussion}.

\section{Preliminaries}\label{section:preliminaries}
For a prime power~$q$ let~$\bF_q$ be the field with~$q$ elements. In a PIR protocol (not necessarily a graph-based one), a dataset $X=(\boldx_1^\top,\ldots,\boldx_n^\top)^\top\in\bF_q^{n\times f}$, which consists on~$n$ \textit{files}~$\{\boldx_i\}_{i=1}^n$, is stored across~$s$ storage servers in a possibly coded manner.
The user wishes to download the file~$\boldx_\phi$, where for the sake of the probabilistic analysis,~$\phi$ is seen as uniformly distributed over~$[n]\triangleq\{1,2,\ldots,n\}$. To this end, the user uses randomness in order to generate queries~$\boldq_1,\ldots,\boldq_s$, one for every server. In turn, server~$i$ replies with~$\bolda_i$, that is a deterministic function of~$\boldq_i$ and the server's content. The protocol is called \textit{$t$-private} if for every subset~$\cT\subseteq [s]$ of size at most~$t$, 
\begin{align*}
    I(\{ \boldq_j \}_{j\in \cT};\phi)=0,
\end{align*}
where~$I$ denotes mutual information. Alternatively, the protocol is~$t$-private if~$\{ \boldq_j \}_{j\in \cT}$ and~$\phi$ are independent. Finally, the \textit{PIR rate} of the system is $f/\sum_{i\in[s]}|\bolda_i|$, i.e., the ratio between the size of the desired data and the amount of downloaded one, both measured in~$\bF_q$ symbols.

In a graph-based replication system every file is replicated multiple times and each one of the copies is stored on a different server. If all files are replicated an identical number of times~$r$, we say that it is an~$r$-replication system, and~$r$ is its replication factor. In a 2-replication system a graph structure arises, in which nodes represent servers, edges represent files, and an edge is incident with a node if the respective file is stored on the respective server. Similarly, in~$r$-replication systems for~$r>2$ an~$r$-uniform hypergraph\footnote{That is, a hypergraph in which all edges contain an identical number of nodes.} structure arises, and in systems where every file is replicated a different number of times, a non-uniform hypergraph arises. Notice that for~$r=2$, a \textit{multigraph}\footnote{A multigraph is a graph in which a certain edge can appear multiple times. Multiple occurrences of the same edge are called \textit{parallel edges}.} might arise, in cases where there exist two servers that share more than one file in common.
While our analysis does not exclude these cases, they result in poor collusion resistance and impede the overall message. Therefore, we restrict our attention to systems in which every two servers store at most one file in common (see Remark~\ref{remark:2servers1file} for further discussion).

Graphs are denoted by~$G=(E,V)$, where~$E=\{e_1,e_2,\ldots\}$ and~$V=\{v_1,v_2,\ldots\}$. Unless otherwise stated, all graphs in this paper are undirected, and hence, an edge is a subset of vertices (subset of size two in ordinary graphs, and of arbitrary size in hypergraphs). For a given graph~$G'$ we denote its set of edges by $E(G')$ and its set of vertices by~$V(G')$. Since graphs represent storage systems in this paper, the terms \textit{node}, \textit{vertex}, and \textit{server} are used interchangeably, and so does the terms \textit{edge} and \textit{file}. 

For a graph~$G$ and a subset~$\cS\subseteq V(G)$ we denote by~$G_\cS$ the subgraph induced by~$\cS$, i.e., the graph which consists of the nodes in~$\cS$ and all the edges in~$E(G)$ that both of their incident nodes are in~$\cS$. 
A \textit{cycle} in~$G$ is a subgraph of~$G$ whose nodes are~$\{v_i\}_{i=0}^{t-1}$ for some~$t$, and whose edges are~$\{ v_i,v_{i+1\bmod t} \}_{i=0}^{t-1}$, and these edges exist also in~$E(G)$.
An edge~$e$ is said to be \textit{incident} with a vertex~$v$, and vice versa, if~$v\in e$. The set of edges in~$E(G)$ that are incident with~$v$ are denoted by $\Gamma_G(v)$, where~$G$ is omitted if clear from context. The \textit{incidence matrix}~$I(G)$ of a graph~$G$ is a~$|V(G)|\times |E(G)|$ binary matrix in which rows correspond to nodes and columns correspond to edges, and an entry contains~$1$ if and only if the respective vertex is incident with the respective edge. In the sequel, the well-known \textit{Breadth First Search} (BFS) algorithm is used repeatedly, in graphs as well as in hypergraphs, and the uninformed reader is referred to~\cite{Cormen}.

In all subsequent protocols, the queries~$\boldq_1,\ldots,\boldq_s$ are vectors in~$\bF_q^n$, i.e., they contain a field element for every file. However, since the servers contain only a portion of the files in the system, the user communicates only their support to the servers. We denote by~$Q$ the~$s\times n$ matrix whose~$i$'th row is~$\boldq_i$ for every~$i\in[s]$, and note that it is a random variable that depends on~$\phi$, and on the randomness at the user. In cases where~$\phi$ is fixed, we denote the matrix of queries by~$Q\vert \phi$.

Since submatrices are used repeatedly, we define the following notation. For a matrix~$A\in\bF^{s\times n}$ and sets~$\cS\subseteq[s]$ and~$\cN\subseteq[n]$, let~$A_{\cS,\cN}$ be the submatrix of~$A$ that consists of the rows in~$\cS$ and the columns in~$\cN$. Further, let~$A_{:,\cN}\triangleq A_{[s],\cN}$ and~$A_{\cS,:}\triangleq A_{\cS,[n]}$. For vectors~$\bolda\in\bF_q^n$ and~$\boldb\in\bF_q^s$ we define~$\bolda_{\cN}$ and~$\boldb_\cS$ analogously. For convenience, we consider the rows and columns of a matrix~$A_{\cS,\cN}$ as indexed by~$\cS$ and~$\cN$, respectively, rather than by~$[|\cS|]$ and~$[|\cN|]$. For example, if~$n=s=4$ and~$\cS=\cN=\{2,3\}$, then~$A_{\cS,\cN}$ is a~$2\times 2$ matrix whose entries are indexed by~$(2,2),(2,3),(3,2),(3,3)$. Since submatrices of~$Q$ are in strong correspondence with subgraphs of~$G$, for every subgraph~$T$ of~$G$ (denoted $T\subseteq G$) we denote~$Q^{T}\triangleq Q_{V(T),E(T)}$, and similarly, for every vector~$\boldv\in\bF_q^s$ we define~$\boldv^{T}\triangleq \boldv_{V(T)}$.

By and large, we use lower-case letters ($a,b,c,\ldots$) to denote scalars, boldface letters ($\bolda,\boldb,\boldc,\ldots$) to denote vectors (all of which are row vectors), capital letters ($A,B,C,\ldots$) to denote matrices or graphs, and calligraphic letters ($\cA,\cB,\cC,\ldots$) to denote sets. Finally, we use the standard notation~$[N,K]_q$ to denote a linear code of length~$N$ and dimension~$K$ over~$\bF_q$.


\subsection{Previous work}
Originally defined in~\cite{OriginalPaper}, the PIR problem has attracted a tremendous amount of research in the past two decades; and due to its tight connection with distributed storage, PIR enjoyed an increasing attention in the past few years. Since a comprehensive summary of previous works is beyond the scope of this paper, we list herein only a partial list of recent contributions, and elaborate on the most relevant ones.

The recent surge of interest in PIR, which addresses the problem from a distributed storage standpoint, includes the reduction of storage overhead by using error correcting codes in~\cite{Eitan} and its improvement in~\cite{TuviAndSimon}; obtaining secrecy by one extra bit in~\cite{OneExtra} and its improvement in~\cite{TuviAndSimon2}; and an extensive line of works regarding achievability and capacity in various scenarios, such as multi-round, multi-message, symmetric,  and with byzantine or colluding servers~\cite{SunAndJafar1,SunAndJafar2,Tajeddine2,CapacityCodedNoCollusion,symmetric,Multimessage,SunAndJafar3}. 
This line of works is a natural extension of an earlier one in the computer science community, which addressed the problem in a more simplistic fashion. Namely, the dataset is assumed to be replicated in its entirety on all servers in the system, and the files are assumed to consist of a single bit. Furthermore, this problem is strongly connected to \textit{locally decodable codes}~\cite{Yekhanin,YekhaninBook}, and has seen a substantial progress recently~\cite{Dvir}.

All of the aforementioned works fall into either one of two extremes in the approach towards PIR. In one, the dataset in its entirety is stored in every server, and in the other it is coded by using an MDS code. The current work addresses a sweet spot between the two, that is strongly motivated by real-world applications~\cite{Hadoop,Cassandra}, as well as a plethora of storage models that were addressed in the past~\cite{FR,SalimFR,Lev1,Lev2,ExpanderCodes}.

Nevertheless, two notions that are relevant to this work were recently addressed in the literature. First, one may consider the special case of graph-based replication in which the degree\footnote{The degree of a node in a graph is the number of edges that are incident with it.} of the nodes in the graph is upper bounded by some parameter. Evidently, this special case is strongly connected to a recent work~\cite{StorageConstrained}, that addressed the general coded PIR question in cases where each server is constrained to contain only a fraction of the entire dataset. Yet, \cite{StorageConstrained} did not impose the particular replication structure that is fundamental to our approach, and more importantly, did not consider collusion. Furthermore, we emphasize that our graph-based approach is highly flexible, in the sense that no constraint is imposed other than every file being replicated on a subset of the servers.

Another notion that was previously studied is that of \textit{collusion patterns}~\cite{Tajeddine3,Patterns}. In this variant, the system must guarantee collusion resistance against specific subsets of servers, rather than any subset up to a certain size. This notion bears some similarity to this work, since one may compel the vertices in these specific sets not to induce a subgraph which infringes privacy in our scheme. However, the approach and the results of these works is substantially different from ours, e.g., since~\cite{Tajeddine3} only discuss coded storage, and~\cite{Patterns} discussed replication of the entire dataset in every server, and disjoint colluding sets.

\section{Replication factor two}\label{section:Replication2}
\subsection{A PIR protocol for 2-replication systems} \label{section:Replication2Subsection}
In this section it is assumed that the replication factor is two, and that every two servers store at most one file in common (see Remark~\ref{remark:2servers1file}), which results in a graph~$G=(V,E)$. The scheme applies for any field~$\bF_q$ with at least three elements. Upon requiring file~$\boldx_{\phi}$, the user randomly chooses a vector~$\boldalpha=(\alpha_i)_{i=1}^n\in(\bF_q^*)^n$, a vector~$\boldgamma=(\gamma_i)_{i=1}^s\in(\bF_{q}^*)^s$, and an element~$h\in\bF_q\setminus\{ 0,1 \}$, all uniformly at random, and defines
\begin{align*}
	Q\triangleq\diag(\boldgamma)\cdot I_{\phi} \cdot\diag(\boldalpha),
\end{align*}
where~$I_{\phi}$ is obtained from~$I(G)$ by replacing the lower~$1$-entry in each column with~$-1$, and then replacing the~$1$-entry in column~$\phi$ by~$h$.

Let~$\boldq_j$, the query for server~$j$, be the $j$-th row of $Q$. Clearly, to upload this row we only need to send the values of its nonzero entries, and hence the total upload complexity is~$2n$. Each node responds with~$\bolda_j=\boldq_j\cdot X$, and therefore the download complexity is~$sf$, and the PIR rate is~$1/s$. Note that node~$j$ can calculate the inner product since the support of~$\boldq_j$ contains only the indices of the files available to it. Upon receiving the information from all~$s$ servers, the user has access to~$QX=\diag(\boldgamma)I_{\phi}\diag(\boldalpha)X$. Then, by multiplying from the left by the matrix~$\diag(\boldgamma)^{-1}$ and by the all ones vector~$\1$, the user get 
\begin{align*}
\1\cdot\diag(\boldgamma)^{-1}\diag(\boldgamma)I_{\phi}\diag(\boldalpha)X=\1\cdot I_{\phi}\diag(\boldalpha)X=(h-1)\alpha_{\phi}\boldx_{\phi},
\end{align*}
and hence~$\boldx_\phi$ can be recovered. We proceed with studying the collusion resistance of the suggested scheme. The following claim is a special case of a more general one that is given in the sequel (Theorem~\ref{theorem:support|phi}). Nevertheless, it is given here in its current form to maintain simplicity and flow, and its proof is sketched.

\begin{proposition}\label{proposition:perfectSecrecy}
	For any set of servers~$\cS\subseteq V$ such that~$G_\cS$ does not contain a cycle, we have that~$I(\{ \boldq_i \}_{i\in\cS};\phi)=0$.
\end{proposition}

\begin{proof}[Proof sketch]
    To prove the claim, we analyze the submatrix of queries that is seen by~$\cS$. For clarity, we omit zero columns from this matrix, as well as columns of weight one, since the latter ones are obviously purely random, and cannot cause leakage of information. Hence, the matrix we analyze is chosen according to the random variable~$Q^{G_\cS}$.
    
    It is evident that every matrix which is chosen according to~$Q^{G_\cS}$ has support which is identical to that of~$I(G)^{G_\cS}$. In what follows we explain why \textit{every} $|V(G_\cS)|\times |E(G_\cS)|$ matrix~$M$ whose support is identical to that of~$I(G)^{G_\cS}$ can be obtained by some choice of~$\boldgamma,\boldalpha$, and~$h$ with identical probability, regardless of the value of~$\phi$. Consequently, this proves that no information regarding~$\phi$ is leaked. 
    
	We calculate~$\Pr(Q^{G_\cS}=M)$ by an iterative process that follows a Breadth First Search (BFS) transversal on~$G_\cS$. Pick an arbitrary~$v_i\in\cS$, and fix the value of the corresponding~$\gamma_i$ (with probability one). Clearly, it follows that $\Pr(\gamma_i \cdot\alpha_j\cdot (I_{\phi})_{i,j}=M_{i,j})=(q-1)^{-1}$ for every~$e_j\in\Gamma_{G_\cS}(v_i)$ regardless of whether or not~$(I_{\phi})_{i,j}$ is the entry of~$I_\phi$ which is multiplied by~$h$. Having the values of~$\alpha_j$ for every~$e_j\in \Gamma_{G_\cS}(v_i)$ fixed, we have that~$\Pr(\gamma_{j'}\cdot\alpha_j\cdot(I_{\phi})_{j',j})=(q-1)^{-1}$ for the same reasons, where~$v_{j'}$ is the other end of edge~$e_j$ (again, regardless of whether or not~$(I_{\phi})_{j',j}$ is the entry of~$I_\phi$ which is multiplied by~$h$). In other words, we have that fixing an entry in~$\boldgamma$ which corresponds to some~$v\in V(G_\cS)$ compels us to fix the values in~$\boldalpha$ which correspond to all of~$\Gamma_{G_\cS}(v)$. In turn, fixing these entries of~$\boldalpha$ compels us to fix the values of~$\boldgamma$ at the other endpoints of the edges in~$\Gamma_{G_\cS}(v)$. Since~$G_\cS$ does not contain a cycle, we may proceed in a BFS fashion and have that every edge-node incidence in~$G_\cS$ reduces the overall probability of obtaining~$M$ by~$(q-1)^{-1}$. Hence, every such matrix~$M$ is obtained with probability~$(q-1)^{-|M|}$, where~$|M|$ is the size of the support of~$M$, and regardless of the value of~$\phi$. Hence, perfect privacy is guaranteed.
\end{proof}

We now turn to study how gracefully the perfect privacy deteriorates if~$\cS$ contains one or more cycles, i.e., how much of~$\phi$'s identity is revealed. 

\begin{proposition}\label{proposition:invertable}
	For any cycle~$C=(V',E')$ in~$G$, any matrix~$M$ in the support of the random variable~$Q^C$ is invertible if and only if~$e_\phi\in E'$.
\end{proposition}

\begin{proof}
	Let~$A\triangleq \diag(\boldgamma_{V'})^{-1} M \diag(\boldalpha_{E'})^{-1}$, and observe that~$\rank(A)=\rank(M)$. If~$\phi\notin E'$, then each column of~$A$ has two nonzero entries~$1$ and~$-1$. Hence,~$\1$ is in its left kernel, and thus~$\rank(A)<c$, where~$c\triangleq |V'|=|E'|$. Moreover, it is an easy exercise to show that any set of~$c-1$ columns of~$A$ are linearly independent, and hence~$\rank(A)=c-1$. 
	
	On the other hand if~$\phi\in E'$, assume without loss of generality that~$A$ is of the form
	\begin{align*}
		A=\begin{pmatrix}
			* &   &        &   &  h\\
			* & * &        &   &   \\
			  & * & \ddots &   &   \\
			  &   & \ddots & * &   \\
			  &   &        & * & -1  
		\end{pmatrix},
	\end{align*}
	where~$*$ denotes a nonzero entry. 
	Then, $\det A=(-1)^{c-1}h\cdot\det A_1- \det A_2$, where~$A_1$ (resp.~$A_2$) is the bottom-left (resp. top-left) $(c-1)\times(c-1)$ submatrix of~$A$. Notice that~$\det A_1$ is the product of all $*$-entries in the sub-diagonal of~$A$, and that~$\det A_2$ is product of all $*$-entries in the main diagonal of~$A$. Hence, since every pair of~$*$-entries in any given column are negations of one another, it follows that~$\det A_1=(-1)^{c-1}\det A_2$. Thus, $\det A=(-1)^{2c-2}h\cdot \det A_2-\det A_2=(h-1)\det A_2\ne 0$.
\end{proof}

\begin{corollary}\label{corollary:T}
	A set~$\cS\subseteq V$ such that~$G_\cS$ contains cycles can narrow down the possible values of~$e_{\phi}$ (and hence, of~$\phi$ itself) to
	\begin{align}\label{equation:T_S}
	\cT=\cT(\cS,\phi)\triangleq \left(\bigcap_{k=1}^\ell E(C_k)\right)\setminus\left(\bigcup_{k=1}^{\ell'}E(C_k')\right),
	\end{align}
	where~$C_1,\ldots,C_\ell$ are all cycles in~$G_\cS$ that contain\footnote{For~$\ell=0$ we formally define~$\bigcap_{k=1}^\ell E(C_k)=E$.}~$e_{\phi}$, and~$C_1',\ldots,C_{\ell'}'$ are all cycles in~$G_\cS$ that do not contain~$e_{\phi}$. 
\end{corollary}

\begin{proof}
	Let~$M$ be the matrix that is seen by~$\cS$; chosen according to the random variable~$Q^{G_\cS}$. By Proposition~\ref{proposition:invertable}, the colluding servers can compute the rank of~$M^C$ for every cycle~$C$ in their induced subgraph, and deduce if~$e_{\phi}\in E(C)$ accordingly. 
\end{proof}

We now show that Corollary~\ref{corollary:T} is in some sense the best that the colluding servers can hope for. Formally, we show that conditioned by~$e_{\phi}\in\cT$,  all respective possible queries are obtained with identical probability. The immediate conclusion is that out of the~$\log n$ protected bits of~$\phi$, the information leakage if a set~$\cS$ collude is precisely~$\log n-\log|\cT|$; or, differently put, all files in~$\cT$ are equally likely. 


To state the main theorem of this paper, whose proof is given in Appendix~\ref{section:MainTheorem}, and of which Proposition~\ref{proposition:perfectSecrecy} is a special case, we require the following definition. For~$\cS\subseteq V$ and~$\cD\subseteq E$, we say that a matrix in~$\bF_q^{|\cS|\times |\cD|}$ is $(\cS,\cD)$-\textit{compatible} with~$G$ ($(\cS,\cD)$-compatible, for short) if its support coincides with that of~$I(G)_{\cS,\cD}$. This definition extends naturally to a subgraph~$T\subseteq G$ where a matrix in~$\bF_q^{|V(T)|\times |E(T)|}$ is said to be~$T$-compatible if it is~$(V(T),E(T))$-compatible.

\begin{theorem}\label{theorem:support|phi}
    For every subgraph~$T\subseteq G$, the support of the random variable~$Q^T\vert\phi$ is the set of all matrices~$A\in\bF_q^{|V(T)|\times|E(T)|}$ such that:
    \begin{itemize}
        \item[(a)]$A$ is~$T$-compatible with~$G$; and
        \item[(b)]for every cycle~$C\subseteq T$,
    \begin{align*}
        \rank(A^C)=\begin{cases}
            |E(C)| & \mbox{if }\phi\in E(C)\\
            |E(C)|-1 & \mbox{if }\phi\notin E(C)
        \end{cases}.
    \end{align*}
    \end{itemize}
    Furthermore, the random variable~$Q^T\vert\phi$ is uniformly distributed on its support. 
\end{theorem}

First, it is evident that the case where~$T$ is acyclic in Theorem~\ref{theorem:support|phi} proves Proposition~\ref{proposition:perfectSecrecy}. Second, we have the following corollary.

\begin{corollary}
    For every set~$\cS\subseteq V$ and every two distinct values~$\phi_1,\phi_2\in[n]$ such that~$\phi_2\in\cT(\cS,\phi_1)$, the servers in~$\cS$ cannot infer if~$\phi=\phi_1$ or~$\phi=\phi_2$.
\end{corollary}
\begin{proof}
    Clearly, it suffices to prove that the random variables~$Q^{G_\cS}\vert (\phi=\phi_1)$ and~$Q^{G_\cS}\vert (\phi=\phi_2)$ are identical, i.e., the same queries are obtained with identical probabilities. Since both random variables are uniformly distributed on their support by Theorem~\ref{theorem:support|phi}, it suffices to prove that their supports are identical. Also by Theorem~\ref{theorem:support|phi}, it suffices to prove that the conditions (a) and (b) coincide in both cases. For~(a) this claim is clear since it does not depend on the value of~$\phi$. For condition~(b), we need to prove that~$\phi_1\in E(C)$ if and only if~$\phi_2\in E(C)$ for every cycle~$C$ in~$G_\cS$, which is precisely the meaning of~$\phi_2\in\cT(\cS,\phi_1)$.
\end{proof}

We now turn to present several choices of the graph~$G$, and the resulting privacy of the PIR schemes. These examples are summarized in Table~\ref{table:Gexamples}.

\begin{example}~
	\begin{enumerate}
		\item Taking~$G$ to be the \textit{Petersen graph} (a~$3$-regular graph with~$10$ nodes,~$15$ edges, and girth~$5$) allows to store~$15$ files on~$10$ servers, $3$ files on each, where any~$4$ servers cannot infer any information regarding~$\phi$. According to the structure of the Petersen graph, at least~$8$ servers are required to infer the~\textit{exact} identity of~$\phi$. The upload complexity is~$30$ field elements, and the download complexity is~$10f$ field elements, i.e., the PIR rate is~$0.1$.
		\item \label{item:bipartite}
		Taking~$G=(\cL\cup\cR,\cV)$ to be the complete bipartite graph, with~$n$ a square integer and $|\cL|=|\cR|=\sqrt{n}$, allows to store~$n$ files on~$2\sqrt{n}$ servers.
		To retrieve a file~$\boldx_{\phi}$, the user downloads~$2\sqrt{n}\cdot f$ field elements. The resulting system ensures perfect privacy against all sets~$\cS\subseteq \cL\cup \cR$ such that either~$|\cS\cap \cL|\le 1$ or~$|\cS\cap \cR|\le 1$, and in particular, all sets of size three.
		\item Graphs of large (constant) girth~$g$ are particularly useful since all sets with at most~$g-1$ nodes are cycle-free, and hence the resulting protocol is~$(g-1)$-private. These can be obtained as incidence graphs of \textit{generalized polygons}~\cite[Table~I]{FR}, of which Item~\ref{item:bipartite} above is a special case. In particular, for prime power~$q$, there exist explicit graphs with degree~$q+1$ with $s\in\{ O(q^2), O(q^3), O(q^5) \}$ (and hence $n\in\{ O(q^3), O(q^4), O(q^6) \}$), where~$g\in \{ 6,8,12 \}$, respectively. The respective download complexities are~$O(n^{2/3})\cdot f$, $O(n^{3/4})\cdot f$, and~$O(n^{5/6})\cdot f$.
		\item Let~$p\ge 5$ be a prime, and let~$m$ be a positive integer. The \textit{Murty graph}~\cite{Murty} is a~$(p^m+2)$-regular graph with~$s=2p^{2m}$ nodes, $n=p^{2m}(p^m+2)$ edges, and girth five. In the resulting system, a database of~$n$ files is stored on~$O(n^{2/3})$ servers, $O(n^{1/3})$ files in each, and ensures perfect privacy against any four colluding servers. To retrieve a file, a user downloads~$O(n^{2/3})\cdot f$ field elements.
		\item Ramanujan graphs (e.g.,~\cite{Ramanujan}) with~$n$ edges and constant degree have girth~$O(\log n)$. Hence, the system is resilient against any~$O(\log n)$ colluding servers, but require download of~$\delta n f$ field elements for some~$\delta \in (0,1)$.
	\end{enumerate}
\end{example}

    \begin{table}[]
        \centering
\begin{tabular}{c|c|c|c|c|c|}
\cline{2-6}
\multicolumn{1}{l|}{}                                                                  & \cellcolor[HTML]{C0C0C0}\textbf{$n$} & \cellcolor[HTML]{C0C0C0}\textbf{$s$} & \cellcolor[HTML]{C0C0C0}\textbf{$t$} & \cellcolor[HTML]{C0C0C0}\textbf{$d$} & \cellcolor[HTML]{C0C0C0}\textbf{PIR rate} \\ \hline
\multicolumn{1}{|c|}{\cellcolor[HTML]{C0C0C0}\textbf{Petersen}}                        & $15$                                 & $10$                                 & $4$                                  & $3$                                  & $\frac{1}{10}$                            \\ \hline
\multicolumn{1}{|c|}{\cellcolor[HTML]{C0C0C0}\textbf{Complete bipartite}}              & Square                               & $2\sqrt{n}$                          & $3$                                  & $\sqrt{n}$                           & $\frac{1}{2\sqrt{n}}$                     \\ \hline
\multicolumn{1}{|c|}{\cellcolor[HTML]{C0C0C0}}                                         & $O(q^3)$                             & $O(q^2)$                             & $5$                                  & $q+1$                                & $O(n^{-2/3})$                             \\ \cline{2-6} 
\multicolumn{1}{|c|}{\cellcolor[HTML]{C0C0C0}}                                         & $O(q^4)$                             & $O(q^3)$                             & $7$                                  & $q+1$                                & $O(n^{-3/4})$                             \\ \cline{2-6} 
\multicolumn{1}{|c|}{\multirow{-3}{*}{\cellcolor[HTML]{C0C0C0}\textbf{Gen. polygons}}} & $O(q^6)$                             & $O(q^5)$                             & $11$                                 & $q+1$                                & $O(n^{-5/6})$                             \\ \hline
\multicolumn{1}{|c|}{\cellcolor[HTML]{C0C0C0}\textbf{Murty}}                           & $p^{2m}(p^m+2)$                      & $2p^{2m}$                            & $4$                                  & $p^m+2$                              & $O(n^{-2/3})$                             \\ \hline
\multicolumn{1}{|c|}{\cellcolor[HTML]{C0C0C0}\textbf{Ramanujan}}                       & Any                                  & $\frac{2n}{d}$                       & $O(\log n)$                          & Constant                             & $\frac{d}{2n}$                            \\ \hline
\end{tabular}
\vspace{0.4cm}
\caption{Different examples for the choice of~$G$ in Section~\ref{section:Replication2}. The parameter~$t$ stands for the guaranteed~$t$-privacy of the system, and~$d$ denotes the fixed degree of the vertices in the graph.}\label{table:Gexamples}
\end{table}

\begin{remark}\label{remark:2servers1file}
    It is evident that the correctness of the scheme and its privacy guarantees hold also in cases where there exist two servers that store more than one file in common. However, in the resulting multigraph, these two servers form a cycle, and hence can collude to infer some information regarding the identity of~$\boldx_\phi$. On the one hand, the system designer may choose to disperse the files while ignoring the aforementioned restriction in order to increase the number of files in the system, at the price of diminishing its privacy guarantees. On the other hand, if the system is designed such that every two servers store at most one file in common, it is clear that~$n\le{ s \choose 2}$.
\end{remark}

\subsection{Bound}\label{section:upperbound}
In this subsection we explore the limitations of PIR protocols for graph-based replication systems by proving a bound on the PIR rate. The resulting bound is particularly powerful for the important family of regular graphs, for which the bound is within a factor of two from the rate in Subsection~\ref{section:Replication2Subsection}. We prove the bound for two-replication systems that provide nontrivial privacy guarantees, namely, the system is at least two-private. In addition, the maximum degree of a vertex in~$G$ is denoted by~$\delta$.

\begin{lemma}
    In every two-private two-replication system the PIR rate is at most~$\frac{\delta}{n}$.
\end{lemma}

\begin{proof}
    Let~$G$ be the induced graph, and let~$\mu_i\ge 0$ be the fraction of~$f$ which is downloaded from server~$i$ by the user. Clearly, it must be that~$\mu_i+\mu_j\ge 1$ for every edge~$\{i,j\}\in E(G)$, since otherwise, servers~$i$ and~$j$ can infer that their mutual file is not required by the user, and hence the system is not two-private. Further, the PIR rate of the system is~$(\1_s\cdot \boldmu^\top)^{-1}$, where~$\1_s$ is the all~$1$'s vector of length~$s$ and~$\boldmu\triangleq(\mu_1,\ldots,\mu_s)$. Hence, an upper bound on the PIR rate of the system is obtained from the optimal solution of the following linear program.
    \begin{align}\label{equation:linearProg}
        \min ~\1_s\cdot \boldmu^\top \mbox{, subject to }I(G)^\top \boldmu^\top \ge \1_n\mbox{ and }\boldmu\ge 0,
    \end{align}
    That is, the inverse of the optimum value of the objective function serves as an upper bound on the PIR rate of the system. The following problem, which is called the \textit{dual} of~\eqref{equation:linearProg}, $\boldeta$ is a vector of~$n$ variables.
    \begin{align}\label{equation:dual}
        \max~\1_n\cdot \boldeta^\top \mbox{, subject to }I(G) \boldeta^\top\le \1_s\mbox{ and }\boldeta\ge 0.
    \end{align}
    According to the primal-dual theory~\cite[Sec.~29.4]{Cormen}, any solution which is feasible for~\eqref{equation:dual} provides a lower bound for~\eqref{equation:linearProg}. It is readily verified that~$\boldeta =\frac{1}{\delta}\cdot \1_n$ is a feasible solution for~\eqref{equation:dual}, and the objective function for this solution equals~$n/\delta$. Therefore, the PIR rate is bounded by~$\delta/n$.
\end{proof}

In cases where~$G$ is a regular graph, which are particularly interesting since they induce systems with balanced storage, the resulting bound equals~$\frac{\delta}{n}=\frac{2\delta}{s\delta}=2/s$. However, the possibility of a considerable rate improvement in highly-unbalanced systems remains widely open.

\section{Arbitrary replication factors}\label{section:largerReplicationFactor}
In this section we consider $r$-replication systems for~$r\ge 2$, which are favored in practice due to their greater resilience to simultaneous failures~\cite{Hadoop,Cassandra}. 
First, for any integer~$r\ge 2$, collusion resistance of~$r-1$ can be attained by a simple scheme that is given in Subsection~\ref{section:rRepANDrm1Col}. Then, we provide another scheme in Subsection~\ref{section:arbitrary}, which guarantees larger collusion resistance by a reduction to the $2$-replication case. The collusion resistance in the latter case will strongly depend on our ability to increase the girth by removing edges from a certain multigraph. 
To simplify the discussion, in this section we alleviate the requirement that every two servers share at most one file in common.

\subsection{Replication factor~$r$ and collusion resistance~$r-1$}\label{section:rRepANDrm1Col}
The user begins by choosing a uniformly random matrix~$V\in\bF_q^{r\times n}$, whose rows sum to~$\bolde_\phi$,
the~$\phi$'th unit vector of length~$n$. Then, the user disperses the~$nr$ symbols of the 
matrix~$V$
to the queries~$\{\boldq_i\}_{i=1}^s$ arbitrarily\footnote{This is possible since~$\sum_{i=1}^s |\boldq_i|=\sum_{i=1}^s |\Gamma(i)|=rn$, where~$|\boldq_i|$ is the length of~$\boldq_i$.}, such that every server that stores a file~$\boldx_j$ receives a unique entry from the~$j$'th column of~$V$.
In turn, the servers respond with the respective linear combinations~$\{\bolda_j=\boldq_j\cdot X\}_{j=1}^s$, and the user computes~$\sum_{i=1}^s\bolda_i=\bolde_\phi\cdot X=\boldx_\phi$. 

It is readily verified that every set of~$r-1$ servers can observe at most~$r-1$ entries in every column of~$V$, which appear entirely random, and hence the resulting scheme is~$r-1$ private. Notice that there is no restriction on the number of files that can be stored in this system, nor there is a restriction on their dispersion.

\subsection{Arbitrary replication factor by reduction}\label{section:arbitrary}
In systems where files might be stored in more than two servers, one can obtain perfect privacy by ``ignoring'' all but two copies of every file that is replicated more than twice, in a sense that will be made clear shortly, and applying the scheme in Section~\ref{section:Replication2}. Observe that choosing which copies to ignore may drastically affect the collusion resistance of the system, since each choice produces a different graph with different cycles. Nevertheless, this observation can in fact contribute to the security of the system by concealing the cycle structure of the resulting graph from an adversary. In what follows we formalize these intuitions and discuss the different aspects of the reduction to the 2-replication scheme.

Evidently, it is natural to consider an $r$-replication system for~$r\ge 2$ (or in fact, any replication system) as a hypergraph, where each file corresponds to a hyperedge. Yet, for our purpose it is often more convenient to consider it as a~\textit{colored multigraph}. That is, instead of considering every file as a hyperedge, which is incident with the nodes that contain it, we consider a multigraph in which every edge carries a label (or a color) in~$[n]$. Then, two servers are connected by an edge with label~$i\in[n]$ if both of them contain a copy of~$\boldx_i$. Clearly, given a hypergraph~$G$, one can easily create the respective colored multigraph~$\hat{G}$ by replacing hyperedge~$i$ with a clique whose edges are labelled by~$i$. Notice that~$\hat{G}$ can be a multigraph (i.e., contain parallel edges) since hyperedges can intersect in more than one node. An illustration of these definitions is given in Figure~\ref{figure:GandGhat}, which also demonstrates the natural notions of a \textit{monochromatic} and \textit{polychromatic} cycles, that will be useful in the sequel. In what follows we use~$G$ and~$\hat{G}$ interchangeably.

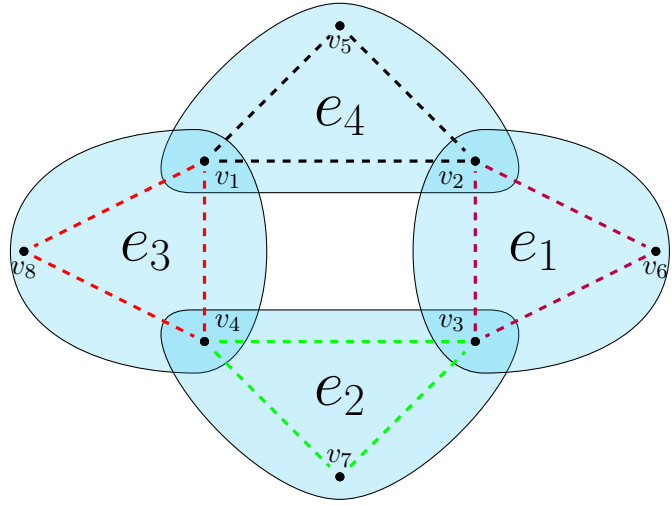
\begin{figure}
\centering
\begin{tikzpicture}[scale=0.6]
    \node (v1) at (6,8) {};
    \node (v2) at (12,8) {};
    \node (v3) at (12,4) {};
    \node (v4) at (6,4) {};
    \node (v5) at (9,11) {};
    \node (v6) at (16,6) {};
    \node (v7) at (9,1) {};
    \node (v8) at (2,6) {};

    \begin{scope}[fill opacity=0.2]
    \filldraw[fill=cyan!70] ($(v1)+(-0.3,-0.7)$) 
        to[out=0,in=180] ($(v2) + (0.3,-0.7)$) 
        to[out=0,in=0] ($(v5) + (0,0.5)$)
        to[out=180,in=180] ($(v1) + (-0.3,-0.7)$);
    \filldraw[fill=cyan!70] ($(v4)+(-0.2,-0.7)$)
        to[out=0,in=0] ($(v1)+(-0.2,0.7)$)
        to[out=180,in=90] ($(v8)+(-0.3,0)$)
        to[out=270,in=180] ($(v4)+(-0.2,-0.7)$);
    \filldraw[fill=cyan!80] ($(v4)+(-0.3,0.7)$) 
        to[out=0,in=180] ($(v3) + (0.3,0.7)$) 
        to[out=0,in=0] ($(v7) + (0,-0.5)$)
        to[out=180,in=180] ($(v4) + (-0.3,0.7)$);
    \filldraw[fill=cyan!70] ($(v3)+(0.2,-0.7)$)
        to[out=180,in=180] ($(v2)+(0.2,0.7)$)
        to[out=0,in=90] ($(v6)+(0.3,0)$)
        to[out=270,in=0] ($(v3)+(0.2,-0.7)$);
    \end{scope}

    \foreach \v in {1,2,...,8} {
        \fill (v\v) circle (0.1);
    }

    \fill (v1) circle (0.1) node [below right] {$v_1$};
    \fill (v2) circle (0.1) node [below left] {$v_2$};
    \fill (v3) circle (0.1) node [above left] {$v_3$};
    \fill (v4) circle (0.1) node [above right] {$v_4$};
    \fill (v5) circle (0.1) node [below] {$v_5$};
    \fill (v6) circle (0.1) node [below] {$v_6$};
    \fill (v7) circle (0.1) node [above] {$v_7$};
    \fill (v8) circle (0.1) node [below] {$v_8$};

    \node at (13.3,6) {\Huge{$e_1$}};
    \node at (9,2.8) {\Huge{$e_2$}};
    \node at (4.7,6) {\Huge{$e_3$}};
    \node at (9,9) {\Huge{$e_4$}};
    
    \draw[red, dashed, very thick, rotate=30] (v1) -- (v8) -- (v4) -- (v1);
    
    \draw[green, dashed, very thick, rotate=30] (v4) -- (v7) -- (v3) -- (v4);
    
    \draw[purple, dashed, very thick, rotate=30] (v2) -- (v6) -- (v3) -- (v2);
    
    \draw[black, dashed, very thick, rotate=30] (v1) -- (v5) -- (v2) -- (v1);
\end{tikzpicture}
\caption{A hypergraph~$G$ (in light blue) and its respective colored multigraph~$\hat{G}$ (in dashed lines). The vertices~$\{v_1,v_2,v_5\}$ contain a monochromatic cycle, but not a polychromatic one. The vertices~$\{ v_1,v_4,v_3,v_2,v_5 \}$ contain a monochromatic cycle, and a polychromatic one.
	}\label{figure:GandGhat}
\end{figure}

Given a replication system with a respective multigraph~$\hat{G}$, it is obvious that the user can choose any two copies of every file, and apply the scheme from Section~\ref{section:Replication2} while ignoring the remaining copies. Formally, for a server~$i$ that stores a copy of~$\boldx_j$ that is chosen to be ignored by the user, the user simply transmits a zero coefficient for~$\boldx_j$, or omits that coefficient altogether. Further, the operation of ignoring all but two copies of every file corresponds to removing all but one of the edges of every color. Obviously, there are potentially many options to choose which edge to keep for every label, and every such choice can be described by a function~$c:[n]\to E(\hat{G})$ such that the edge~$c(i)$ is labelled by~$i$, for every~$i\in[n]$. For any such~$c$, let~$\hat{G}_c$ be the result of keeping the edges~$\{ c(i) \}_{i\in[n]}$, and removing the remaining ones. It is readily verified that the resulting scheme guarantees perfect privacy against colluding sets that do not contain a cycle in~$\hat{G}_c$.

Clearly, if one can choose the file dispersion in the system as one pleases, then it is possible to first choose the dispersion of only two copies of each file, so that the resulting graph~$G'$ has a certain girth. Then, the remaining copies can be dispersed arbitrarily, and the PIR scheme is performed with respect to the function~$c$ that~$c(i)\in E(G')$ for every~$i$. However, if~$\hat{G}$ is \textit{given} to the user, finding a function~$c$ such that~$\hat{G}_c$ has a large girth requires more care.

For a given~$\hat{G}$ one can choose~$c$ at random. In spite of not having any clear minimum girth guarantee, this approach has the extra benefit of concealing the cycle structure from an adversary. For a given integer~$g$, a function~$c$ such that~$\hat{G}_c$ has girth~$g$, if exists, can be found be deciding the feasibility of the following $\{0,1\}$-program. In this program, for~$i\in[n]$ let~$E_i$ be the set of all~$2$-subsets~$\{a,b\}$ of~$[s]$ such that there exists an edge~$\{a,b\}$ labelled by~$i$.

\begin{itemize}
\item \textbf{Objective}: None.
    \item \textbf{Variables}: $\{ x_{i,\{a,b\}}~\vert~i\in[n]\mbox{ and }\{a,b\}\in E_i\}$.
    \item \textbf{Constraints}:
    \begin{itemize}
        \item $\sum_{\{a,b\}\in E_i}x_{i,\{a,b\}}=1$ for all~$i\in[n]$.
        \item $\sum_{i\vert \{a,b\}\in E_i}x_{i,\{a,b\}}\le 1$ for every~$\{a,b\}$ such that there exists at least one edge~$\{a,b\}$ in~$\hat{G}$.
        \item $\sum_{i\vert \{a,b\}\in E_i}x_{i,\{a,b\}}+\sum_{i\vert \{b,c\}\in E_i}x_{i,\{b,c\}}+\sum_{i\vert \{c,a\}\in E_i}x_{i,\{c,a\}}\le 2$, for every~$a,b,c\in[s]$ that contain at least one triangle in~$\hat{G}$.\\
         $\vdots$
        \item $\sum_{j=1}^g \sum_{i\vert \{a_j,a_{(j+1)\bmod g}\in E_i\}}x_{i,\{ a_j,a_{(j+1)\bmod g} \}}\le g-1$ for every~$a_0,\ldots,a_{g-1}\in[s]$ that contain at least one~$g$-cycle in~$\hat{G}$.
    \end{itemize}
\end{itemize}
Clearly, the first set of constraints guarantees that exactly one edge is chosen for every file~$i\in[n]$. The second set of constraints guarantees that the resulting choice does not contain~$2$-cycles, the next set guarantees that there are no triangles, and so on. Finally, we note that while solving this system for a general~$g$ is NP-hard, the special case~$g=2$ reduces to finding a maximum matching in a bipartite graph, a problem that can be solved efficiently.

\section{Graph-based coding -- Reducing the storage overhead at improved PIR rates}\label{section:graphBasedCoding}
This section discusses storage systems in which every file is similarly stored on a small number of servers, but replication is generalized to arbitrary encoding. Hence, when employing an~$[N,K]_q$ code with rate larger than~$1/2$ (i.e., $K/N>1/2)$, we obtain an improvement over previous schemes in terms of storage overhead. Furthermore, it is shown that the resulting PIR rate is improved whenever~$N-K>1$. However, the (coded) file dispersion must follow a certain structure, and the resulting collusion patterns are in correspondence with \textit{polychromatic cycles} (see Subsection~\ref{section:arbitrary} and Figure~\ref{figure:GandGhat}), as will be explained next.
Finally, we note that the scheme in this section is loosely inspired by ideas from~\cite{CamillasGang} and~\cite{David}.

Essentially, in the scheme of Section~\ref{section:Replication2}, every file~$\boldx_i$ is coded by using a repetition code of length~$2$ over the alphabet~$\bF_q^f$. Then, every symbol of the resulting codeword is stored on a different server. The scheme which is presented in this section generalizes this concept by employing codes other than the repetition code. 

For integers~$N$ and~$K$ 
let~$G\in \bF_q^{K\times N}$ be a generator matrix of an~$[N,K]_{q}$ MDS code~$\cD$. Consider every file~$\boldx_i$ as an $(f/K)\times K$ matrix~$(\boldx_{i,1}^\top,\ldots,\boldx_{i,K}^\top)$ over~$\bF_{q}$, and let $(\boldx_{i,1}^\top,\ldots,\boldx_{i,K}^\top)\cdot G \triangleq (\boldy_{i,1}^\top,\ldots,\boldy_{i,N}^\top)$, where the vectors~$\{ \boldy_{i,j} \}_{j=1}^N$ are called the \textit{codeword symbols} of~$\boldx_i$. Let~$\cL_1,\ldots,\cL_N\subseteq [s]$ be disjoint nonempty subsets whose union is~$[s]$ (and hence we must have~$N\le s$). Then, for every~$i\in[n]$, disperse the~$N$ codeword symbols~$\boldy_{i,1},\ldots,\boldy_{i,N}$ to the servers such that for every~$j\in[N]$, the codeword symbol~$\boldy_{i,j}$ is in exactly one server which belong to~$\cL_j$. For example, one can think of a system in which the servers are partitioned to three disjoint subsets; the servers in the first subset contain the first halves of all files, the servers in the second contain the other half, and the servers in the third contain the sums of the two halves (see Example~\ref{example:codedEx1} and Example~\ref{example:codedEx2} which follow).

The above coding scheme gives rise to an~$N$-uniform $N$-partite hypergraph in the following manner. Let~$[s]$ be the set of vertices, and define hyperedges~$e_1,\ldots,e_n$, such that~$e_i$ contains all servers that store either one of~$\boldy_{i,1},\ldots,\boldy_{i,N}$. It is evident that the edges are of size~$N$, and that the~$N$ parts of the hypergraph are the sets~$\cL_1,\ldots,\cL_N$. Let~$G$ be this hypergraph, and let~$\hat{G}$ be its respective colored multigraph, as described in Subsection~\ref{section:arbitrary}.

We begin by presenting the PIR protocol for the special case~$N-K=K$, and later extend it to other parameters by operating in \textit{rounds}. Begin by choosing~$\boldalpha \in(\bF_q^*)^n, \boldgamma\in(\bF_q^*)^s$, and~$h\in\bF_q\setminus\{0,1\}$ uniformly at random, and pick an arbitrary subset~$\cK\subseteq[N]$ of size~$K$. Then, for every~$m\in[N]$, a server~$j\in[s]$ which belongs to~$\cL_m$ receives the following query.
\begin{align}\label{equation:codingQuery}
    (\boldq_j)_t=\begin{cases}
    \gamma_j \cdot \alpha_t\cdot h^{\delta(t,m)} & \mbox{ if }j\mbox{ contains a codeword symbol of~$\boldx_t$}\\
    0 & \mbox{else}
    \end{cases},
\end{align}
where~$\delta(t,m)$ is a Boolean indicator for the event ``$m\in\cK$ and~$t=\phi$''. Namely, the user transmits to server~$j$ the part of the vector~$\gamma_j\cdot \boldalpha$ that is relevant to it, where arbitrary~$K$ servers that store a codeword symbol of~$\boldx_\phi$ are having the~$\phi$'th entry of $\gamma_j\cdot \boldalpha$ multiplied by~$h$. In turn, a server~$j$ in~$\cL_m$, which stores~$\{ \boldy_{\ell,m}\vert \ell \in \cL\}$ for some~$\cL\subseteq [n]$, responds with~$\bolda_j\triangleq \sum_{\ell\in \cL}(\boldq_j)_\ell\cdot \boldy_{\ell,m}$. Having the responses~$\{ \bolda_i \}_{i=1}^s$, the user composes the following matrix.
\begin{align*}
    \left( \sum_{j\in\cL_1}\gamma_j^{-1}\bolda_j^\top,\ldots,\sum_{j\in\cL_N}\gamma_j^{-1}\bolda_j^\top \right)&=\underbrace{\sum_{j=1}^n\alpha_j(\boldy_{j,1}^\top,\ldots,\boldy_{j,N}^\top)}_{\triangleq Y}+\;\bolde,
\end{align*}
where for~$m\in[N]$, the~$m$'th column of~$\bolde$ is
\begin{align*}    
    (\bolde)_m&=\begin{cases}
    \alpha_\phi(h-1)\boldy_{\phi,m} &\mbox{if }m\in \cK\\
    0 & \mbox{else}
    \end{cases}.
\end{align*}
Now, it is evident that every row in the matrix~$Y$ is a codeword in~$\cD$, whose minimum distance is~$N-K+1$. Therefore, since~$\bolde$ has at most~$K$ nonzero columns, and since~$K=N-K$, a decoding algorithm\footnote{Notice that the ``error values'' are in prescribed positions, and hence, an \textit{erasure correction} algorithm suffices.} for~$\cD$ can extract~$\bolde$ from the matrix that was composed by the user. At this point the user has obtained~$\{ \boldy_{\phi,m} \}_{m\in \cK}$, that are sufficiently many codeword symbols of~$\boldx_\phi$ in order to retrieve it. Therefore, the PIR rate of this scheme is~$\frac{f}{s\cdot (f/K)}=\frac{K}{s}=\frac{N-K}{s}$. The proof of privacy will be given after the general description. 

Notice that in the above scheme, $N-K$ codeword symbols of~$\boldx_\phi$ are obtained, while~$K$ many of those are sufficient to retrieve~$\boldx_\phi$. However, in cases where~$N-K<K$, the scheme will not be successful, and in cases where~$N-K>K$, the resulting scheme will not be exploited to its full potential. 

Therefore, to address cases in which~$K\ne N-K$, we retrieve \textit{multiple} files in \textit{rounds}, a standard practice in the PIR literature (e.g.,~\cite{CamillasGang,David}). That is, we assume that the user wishes to download~$\boldx_{\phi_1},\ldots,\boldx_{\phi_b}$ privately for some~$b\ge 1$, and the protocol operates in~$r\ge 1$ rounds. In each round, the user sends a query to every server, and receives responses from all servers. Specifically, we choose~$b$ and~$r$ so that~$Kb=r(N-K)$, i.e., $r\triangleq\frac{LCM(K,N-K)}{N-K}$ and $b\triangleq \frac{LCM(K,N-K)}{K}$. Prior to executing these rounds, the user fixes the following subsets of~$[N]$
\begin{align}\label{equation:JsNonSystematic}
    J^{(1)} &= J^{(1,1)}\cup J^{(1,2)}\cup \ldots \cup J^{(1,b)}\nonumber\\
    J^{(2)} &= J^{(2,1)}\cup J^{(2,2)}\cup \ldots \cup J^{(2,b)}\nonumber\\
    &\vdots\nonumber\\
    J^{(r)} &= J^{(r,1)}\cup J^{(r,2)}\cup \ldots \cup J^{(r,b)},
\end{align}
such that in every row, the sets in the union are pairwise disjoint, such that~$|J^{(i)}|=N-K$ for every~$i\in[r]$, and such that~$|\cup_{i=1}^sJ^{(i,j)}|=K$ for every~$j\in[b]$. Intuitively, for~$j\in[b]$ and~$i\in[r]$, the set~$J^{(j,i)}$ contains the indices of the codeword symbols of~$\boldx_{\phi_j}$ that are retrieved during round~$i$. The choice of such sets is easy, and is illustrated in Appendix~\ref{appendix:omitted}. 

In each round~$i$ the user executes the aforementioned protocol (for the case~$K=N-K$), where~$J^{(i)}$ is used in lieu of the set~$\cK$. That is, the queries are defined as in~\eqref{equation:codingQuery}, with the difference that $\delta(t,m)$ is a Boolean indicator for the event~``there exists~$j\in[b]$ such that~$t=\phi_j$ and $m\in J^{(i,j)}$''. Having obtained the responses from all servers in round~$i$, the user computes
\begin{align*}
    \left( \sum_{j\in\cL_1}\gamma_j^{-1}\bolda_j^\top,\ldots,\sum_{j\in\cL_N}\gamma_j^{-1}\bolda_j^\top \right)&=\underbrace{\sum_{j=1}^n\alpha_j(\boldy_{j,1}^\top,\ldots,\boldy_{j,N}^\top)}_{\triangleq Y}+\;\bolde',
\end{align*}
where for~$m\in[N]$, the~$m$'th column of~$\bolde'$ is
\begin{align*}    
    (\bolde')_m&=\begin{cases}
    \alpha_{\phi_j}(h-1)\boldy_{\phi_j,m} &\mbox{if }m\in J^{(i,j)}\\
    0 & \mbox{else}
    \end{cases}.
\end{align*}
Since~$|J^{(i)}|=N-K$, a decoding algorithm on the matrix~$Y$ can extract the values of~$\bolde'$. Hence, according to the structures of the sets in~\eqref{equation:JsNonSystematic}, it follows that by the end of the~$r$'th round, the user has obtained the~$K$ codeword symbols~$\{ \boldy_{\phi_j,m} \}_{m\in \cup_i J^{(i,j)}}$ of~$\boldx_{\phi_j}$ for every~$j\in[b]$, and hence all the files~$\{ \boldx_{\phi_j} \}_{j=1}^b$ can be retrieved. The resulting PIR rate is
\begin{align*}
	\frac{bf}{s\cdot (f/K)\cdot r}=b\cdot \frac{K}{sr}=\frac{r(N-K)}{K}\cdot \frac{K}{sr}=\frac{N-K}{s}.
\end{align*} 

\begin{remark}
	Roughly speaking, the scheme which is described in Section~\ref{section:Replication2} is as a special case of the one in this section, where~$K=1$, ~$N=2$, and~$\cD\triangleq \{ (x,-x)\vert x \in \bF_q \}$, and the resulting rate is indeed~$\frac{N-K}{s}=\frac{1}{s}$. However, further simplification is possible for this particular choice of~$\cD$, since the process of extracting the error vector~$\bolde$ reduces to multiplying by~$\1$ from the left. Hence, the partitioning of the servers to subsets~$\{ \cL_j \}_{j=1}^N$ is not required.
\end{remark}




\begin{proposition}\label{proposition:colludingHyper}
	A set~$\cS\subseteq V$ that contains no polychromatic cycles in~$\hat{G}$ gains no information about~$\phi_1,\ldots,\phi_b$.
\end{proposition}


\begin{proof}
    For~$\cS$ that does not contain a polychromatic cycle, let~$\cR\subseteq[n]$ be the set of hyperedges in~$G$ that have two or more vertices in~$\cS$. Similar to Proposition~\ref{proposition:perfectSecrecy}, we analyze the matrix which is chosen according to the random variable~$Q_{\cS,\cR}$. Clearly, every matrix which is chosen according to~$Q_{\cS,\cR}$ is~$(\cS,\cR)$-compatible with~$G$, and we show that the inverse is also true.
    
    Let~$M\in\bF_q^{|\cS|\times|\cR|}$ be a matrix which is~$(\cS,\cR)$-compatible with~$G$. Fix some~$v_i\in\cS$ as the starting point of the BFS algorithm, and choose an arbitrary value for~$\gamma_i$ (with probability~$1$). Once~$\gamma_i$ is fixed, it is evident that~$\Pr(\gamma_i\cdot \alpha_j \cdot h^\delta=M_{i,j})=(q-1)^{-1}$ for every hyperedge~$e_j$ that is incident with~$v_i$ regardless of the value of the Boolean indicator~$\delta$.
    Notice that the only mutual element of these hyperedges is~$v_i$, since otherwise, a polychromatic cycle of length two would exist in~$\hat{G}$. Therefore, once~$\alpha_j$ is fixed for such a hyperedge~$e_j$, we have that~$\Pr(\gamma_\ell\cdot \alpha_j\cdot h^\delta=M_{\ell,j})=(q-1)^{-1}$ for every~$\ell$ such that~$v_\ell\in e_j\cap \cR$, again, regardless of~$\delta$.
    Proceeding in a BFS fashion, we have that each node-hyperedge incidence reduces the overall probability of obtaining~$M$ by a multiplicative factor of~$(q-1)^{-1}$. Since~$\cS$ does not contain a polychromatic cycle, no discrepancy is encountered, which concludes the proof.
\end{proof}

\begin{example}\label{example:codedEx1}
    Consider~$s=12$, and let~$\cD$ be the parity code~$\{ (x,y,x+y)\vert x,y\in\bF_q \}$, and hence~$N=3$ and~$K=2$. Also, let~$\cL_1=\{1,\ldots,4\}$, $\cL_2=\{5,\ldots,8\}$, and~$\cL_3=\{9,\ldots,12\}$. Consider the following~$16$ hyperedges.
    \begin{align*}
        \{1,5,9\} && \{2,5,10\} && \{3,5,11\} && \{4,5,12\}\\
        \{1,6,10\} && \{2,6,11\} && \{3,6,12\} && \{4,6,9\}\\
        \{1,7,11\} && \{2,7,12\} && \{3,7,9\} && \{4,7,10\}\\
        \{1,8,12\} && \{2,8,9\} && \{3,8,10\} && \{4,8,11\}
    \end{align*}
    It is readily verified that every two distinct edges intersect in at most one node, and hence, there are no polychromatic cycles of length~$2$. The resulting system is~$2$-private, has storage overhead~$1.5$, and its PIR rate is~$1/12$.
\end{example}

\begin{example}\label{example:codedEx2}
    Generalizing the previous example, let~$s$ be any integer divisible by~$3$, let~$\cD$ be the parity code, and let $\cL_1=\{1,\ldots,s/3\}$, $\cL_2=\{ s/3+1,\ldots,2s/3\} $, and $\cL_3=\{ 2s/3+1,\ldots,s \}$. Let~$\cM_1,\ldots,\cM_{s/3}$ be \emph{edge-disjoint maximum matchings\footnote{Recall that a matching is a subset of disjoint edges. A \textit{maximal} matching is a matching such that any edges that is added to it violates the disjointness of its edges. A \textit{maximum} matching is a matching of the largest possible cardinality. It is readily verified that a complete bipartite graph~$K_{m,m}$ contains~$m$ disjoint maximum matchings.}} in a complete bipartite graph~$H$ whose one side is~$\cL_2$, and the other is~$\cL_3$. Notice that~$|\cM_i|=s/3$ for every~$i$, and consider the following hyperedges.
    \begin{align*}
        \left\{\{ 1,a,b \} \vert \{a,b\}\in \cM_1 \right\} && \left\{\{ 2,a,b \} \vert \{a,b\}\in \cM_2 \right\} && \ldots &&
        \left\{\{ s/3,a,b \} \vert \{a,b\}\in \cM_{s/3} \right\}
    \end{align*}
    We claim that any two of the above hyperedges intersect in at most one node. Assuming otherwise we have $|\{a_1,a_2,a_3\}\cap\{b_1,b_2,b_3\}|=2$ for some integers~$a_i$ and~$b_i$. If~$a_1=b_1$, it follows that the edges~$\{a_2,a_3\}$ and~$\{b_2,b_3\}$ in~$H$ share a vertex, even though they both belong to~$\cM_{a_1}$, a contradiction. If~$a_1\ne b_1$, it follows that the matchings~$\cM_{a_1}$ and~$\cM_{b_1}$ both contain the edge~$\{a_2,a_3\}=\{b_2,b_3\}$, another contradiction.
    
    Therefore, the resulting system is~$2$-private, accommodates~$n=s^2/9$ files, incurs storage overhead of~$1.5$, and has PIR rate of~$1/s$. For comparison, considering the full graph on~$s$ nodes and applying the scheme in Section~\ref{section:Replication2} provides a $2$-private system with~$n=(s^2+s)/2$ files, storage overhead~$2$, and comparable PIR rate~$1/s$.
\end{example}

\section{Discussion and open questions}\label{section:Discussion}
In this paper we initiated a study of private information retrieval for a specific storage model that is widely used in practice, and widely studied in theoretical research. In order to improve our understanding of this model, and in order to improve its applicability to real-world systems, we suggest the following research directions. 

\begin{enumerate}
    \item Close the gap between achievable PIR rate in Subsection~\ref{section:Replication2Subsection} and the upper bound in Subsection~\ref{section:upperbound}.
    \item Improve the collusion resilience in systems with arbitrary replication factors.
    \item Construct families of dense graphs in which~$\cT(\cS,\phi)$~\eqref{equation:T_S} is large for every~$\cS\subseteq[s]$ and every~$\phi$.
    \item Study graceful degradation for replication factors larger than two.
    \item Find PIR schemes for $2$-replication systems that guarantee collusion resistance against cycles, and are nontrivial (i.e., download less than the entire dataset).
\end{enumerate}

\section*{Acknowledgments}
The work of Itzhak Tamo was supported in part by Israel Science Foundation (ISF) Grant 1030/15 and NSF-BSF Grant 2015814. The work of Eitan Yaakobi was supported in part by Israel Science Foundation (ISF) grant 1817/18. The work of Netanel Raviv was supported in part by
the postdoctoral fellowship of the Center for the Mathematics of Information (CMI) in the California Institute of
Technology.

\appendices

\section{Proof of the main theorem}\label{section:MainTheorem}
The proof of Theorem~\ref{theorem:support|phi} requires two auxiliary lemmas (Lemma~\ref{lemma:technical} and Lemma~\ref{lemma:laste}), and then is proved in two parts (Lemma~\ref{lemma:support|phi1} and Lemma~\ref{lemma:support|phi2}).

\begin{lemma}\label{lemma:technical}
	Let~$C\subseteq G$ be a cycle with~$c$ edges, and let~$M\in\bF_q^{c\times (c-1)}$ be a matrix which is~$(V(C),E(C)\setminus \{j\})$-compatible, where~$j$ is the maximum index of an edge in~$E(C)$. Then, there exist precisely~$q-1$ vectors~$\bolda\in\bF_q^c$ such that~$M'\triangleq(M\vert\bolda)\in\bF_q^{c\times c}$ is~$(V(C),E(C))$-compatible and~$\rank(M')=c-1$.
\end{lemma}
\begin{proof}
	First, observe that since~$C\setminus\{j\}$ is a tree, and since~$M$ is~$(V(C),E(C)\setminus\{j\})$-compatible with~$G$, it follows that~$\rank M=c-1$. Hence, the added vector~$\bolda$ must be in~$\colspan (M)$, i.e.,
	\begin{align}\label{equation:a}
	\bolda = \sum_{{k}\in E(C)\setminus\{ j \}}m_{k}\boldc_{k},
	\end{align}
	where the~$\boldc_{k}$'s are the columns of~$M$ and the~$m_{k}$'s are coefficients from~$\bF_q$. Furthermore, since~$M'$ must be compatible with~$G$, the column~$\bolda$ must contain nonzero entries precisely in row~$i_1$ and row~$i_2$, that correspond to the two vertices incident with edge~$j$. Hence, since each row~$k\in V(C)\setminus\{i_1,i_2\}$ of~$M$ contains precisely two nonzero entries in some columns~$k_1$ and~$k_2$, it follows that intersecting the column span of~$M$ with~$N_k\triangleq\{ \boldx=(x_i)_{i=1}^c\in\bF_q^{c}\vert x_{k}=0  \}$ reduces the degrees of freedom in~\eqref{equation:a} by~$1$, since it renders any one of~$\{m_{k_1},m_{k_2}\}$ to be a linear function of the other. Therefore, 
	\begin{align*}
	\dim\left(X\right)&=(c-1)-(c-2)=1,\mbox{ where}\\
	X&\triangleq \colspan(M)\bigcap \left(\bigcap_{k\in V(C)\setminus\{i_1,i_2  \} }N_k\right).
	\end{align*}
	Since any nonzero vector in~$X$ is a suitable candidate for~$\bolda$, the claim follows.
\end{proof}

\begin{lemma}\label{lemma:laste}
    If an edge~$e\in E(G)$ is on a cycle in~$G$, then there exists a BFS ordering of~$E(G)$ for which~$e$ is a back edge.
\end{lemma}
\begin{proof}
    Denote~$e_\phi=\{ v_f,v_g \}$ and choose~$v_d\in V(G)$ which maximizes~$\dist(v_g,v_d)$, where distance between two vertices is defined as the number of edges in the shortest path between them. Without loss of generality, assume that~$\dist(v_g,v_d)\ge \dist(v_f,v_d)$, and consider a BFS run which begins at~$v_d$. Partition~$V(G)$ to layers~$L_1,L_2,\ldots$ according to their distance from~$v_d$, and recall that edges inside each layer are always back edges. Hence, if~$e_\phi$ is inside a layer, we are done. Otherwise, assume that~$v_f$ is in~$L_i$ for some~$i$, and hence~$v_g$ is in~$L_{i+1}$. Since~$e_\phi$ is on a cycle, there exists another edge~$e'$ from a node~$v'\in L_i$ to~$v_g$. Hence, in cases where~$v'$ pops out of the queue before~$v_f$, $e_\phi$ will indeed be a back edge. It is readily verified that the order of insertion of discovered vertices in the same layer is arbitrary, and hence there exists a BFS run in which~$v'$ predates~$v_f$, and the claim follows.
\end{proof}

We now turn to prove Theorem~\ref{theorem:support|phi} in two parts.

\begin{lemma}\label{lemma:support|phi1}
    For every subgraph~$T\subseteq G$, the support of the random variable~$Q^T\vert\phi$ is the set of all matrices~$A\in\bF_q^{|V(T)|\times|E(T)|}$ such that:
    \begin{itemize}
        \item[(a)]$A$ is~$T$-compatible with~$G$; and
        \item[(b)]for every cycle~$C\subseteq T$
    \begin{align*}
        \rank(A^C)=\begin{cases}
            |E(C)| & \mbox{if }\phi\in E(C)\\
            |E(C)|-1 & \mbox{if }\phi\notin E(C)
        \end{cases}.
    \end{align*}
    \end{itemize}
\end{lemma}
\begin{proof}
    For simplicity assume that~$2\vert q$, but other cases can be proved similarly. By the definition of~$Q^T\vert\phi$, it is evident that~(a) is necessary, and according to Proposition~\ref{proposition:invertable}, it follows that (b) is necessary. In what follows, it is shown that (a) and (b) are also sufficient. To this end, let~$A\in \bF_q^{|V(T)|\times|E(T)|}$ be a matrix which satisfies (a) and (b), and it is shown that there exists a choice of~$\boldalpha,\boldgamma,$ and~$h$ for which~$Q^T\vert\phi$ produces~$A$.
    
    Consider a BFS run on~$T$, and number~$V(T)$ and~$E(T)$ according to their discovery times. That is, let~$v_1,\ldots,v_{|V(T)|}$ be the vertices of~$T$ sorted by their discovery times, and let~$e_1,\ldots,e_{|E(T)|}$ be the edges of~$T$ sorted by their discovery times. Also, assume that if~$e_\phi\in E(T)$, and~$e_\phi$ closes a cycle, then it is a back edge (see Lemma~\ref{lemma:laste}).
    The values of~$\boldalpha,\boldgamma$, and~$h$ which produce~$A$ are determined according to this BFS ordering, as follows.
    
    First, fix an arbitrary value in~$\bF_q^*$ for~$\gamma_1$. Then, since~$v_1$ is incident with the edges~$e_1,\ldots,e_{|\Gamma(v_1)|}$, we fix the values of~$\alpha_1,\ldots,\alpha_{|\Gamma(v_1)|}$ as~$\alpha_i\triangleq A_{v_1,e_i}/\gamma_1, i\in\{1,\ldots,|\Gamma(v_1)|\}$. Then, for $v_2,\ldots,v_{|\Gamma(v_1)|+1}$, that are the end vertices of~$e_1,\ldots,e_{|\Gamma(v_1)|}$, respectively, we fix~$\gamma_i=A_{v_i,e_{i-1}}/\alpha_{i-1}, i\in\{ 2,\ldots,|\Gamma(v_1)|+1 \}$. If~$e_\phi$ is not on a cycle in~$T$, and~$e_\phi$ happens to be, say,~$e_1$, then we can obviously choose~$\alpha_2\triangleq A_{v_2,e_1}/(\gamma_1\cdot h)$, where~$h$ is arbitrary (the case where~$e_\phi$ lies on a cycle is treated in the sequel). Clearly, this process goes on unhindered as long as a back edge is not discovered. 
    
    Once a back edge~$e_b=\{ v_c,v_d \}, b\ne \phi$ is discovered, we have that~$\gamma_c,\gamma_d$ were already determined in earlier stages of the algorithm. Hence, we ought to show that there exists~$\alpha_b$ for which
    \begin{align}\label{equation:alphab}
        \alpha_b&=\frac{A_{v_c,e_b}}{\gamma_c}\mbox{, and}&\alpha_b&=\frac{A_{v_d,e_b}}{\gamma_d}.
    \end{align}
    To this end, let~$C$ be a cycle which is discovered in whole when~$e_b$ is discovered and let~$c$ be its number of edges. Further, let~$M\triangleq A^{C\setminus\{e_b\}}$, i.e., the partial matrix of~$A$ which corresponds to the subgraph~$C\setminus\{e_b\}$. Similarly, let~$N\triangleq \diag(\boldgamma_{V(C)}) I^{C\setminus\{e_b\}} \diag (\boldalpha_{E(C)\setminus\{e_b\}})$ be the matrix which corresponds to the choice of entries in~$\boldgamma$ and~$\boldalpha$ up until~$e_b$ is discovered. By the correctness of the algorithm so far, it follows that~$M=N$. Moreover, both~$M$ and~$N$ are~$(V(C),E(C)\setminus\{j\})$-compatible, and by the definition of~$A$, the submatrix~$A^{C}$ is~$C$-compatible, and its rank is~$c-1$. According to Lemma~\ref{lemma:technical} there exist precisely~$(q-1)$ columns $\boldc_1,\ldots,\boldc_{q-1}$ that extend~$M$ (and also~$N$) to a~$C$-compatible matrix of rank~$c-1$, one of which is~$A^C$. Further, it is evident that the matrix $\diag(\boldgamma_{V(C)}) I^{C} \diag (\boldalpha_{E(C)})$, for \textit{any} of the $(q-1)$ possible values of~$\alpha_b\in\bF_q^*$, results in a~$C$-compatible matrix of rank~$c-1$ as well. Therefore, there exists a 1-1 correspondence between the possible values of~$\alpha_b$ and~$\boldc_1,\ldots,\boldc_{q-1}$. Since one of~$\boldc_1,\ldots,\boldc_{q-1}$ is the actual~$e_b$'th column of~$A^C$, it follows that there exists a unique value of~$\alpha_b\in\bF_q^*$ which satisfies~\eqref{equation:alphab}.
    
    If~$e_\phi$ lies on a cycle~$C'$ in~$T$, we denote~$e_\phi\triangleq\{v_f,v_g\}$. Since~$e_\phi$ is a back edge, we have that~$\gamma_g$ and~$\gamma_f$ were determined in earlier steps of the algorithm. Hence, we must find~$\alpha_\phi\in\bF_q^*$ and~$h\in\bF_q\setminus\{0,1\}$ for which
    \begin{align}
        h\gamma_g\alpha_\phi&=A_{v_g,e_\phi}\label{equation:alphaphi1}\\
        \gamma_f\alpha_\phi&=A_{v_f,e_\phi}.\label{equation:alphaphi2}.
    \end{align}
    Clearly, the choice~$\alpha_\phi\triangleq A_{v_f,e_\phi}/\gamma_f$ satisfies~\eqref{equation:alphaphi2}, and consequently,~$h\triangleq \frac{A_{v_g,e_\phi}}{\gamma_g\alpha_\phi}$ satisfies~\eqref{equation:alphaphi1}. We are only left to show that this value for~$h$ is neither~$0$ nor~$1$. First, it is obviously nonzero as a product of nonzero terms. Second, if~$h=1$ happens to be the answer, we have by Proposition~\ref{proposition:invertable} that~$A^{C'}$ is rank-deficient, in contradiction with condition~(b).
\end{proof}

\begin{lemma}\label{lemma:support|phi2}
    For every~$T\subseteq G$, the random variable~$Q^T\vert\phi$ is uniformly distributed on its support.
\end{lemma}
\begin{proof}
    Let~$A$ be a matrix in the support of~$Q^T\vert\phi$. By following the proof of Lemma~\ref{lemma:support|phi1}, we have that once~$\gamma_1$ is fixed, and as long as a back edge is not discovered, every edge-node incidence reduces the overall probability of obtaining~$A$ by~$(q-1)^{-1}$. In addition, every back edge which is not~$e_\phi$ reduces the probability of obtaining~$A$ by~$(q-1)^{-1}$ due to~\eqref{equation:alphab}, instead of by~$(q-1)^{-2}$ for tree edges\footnote{An edge which is not a back edge in a BFS ordering is called a tree edge.}. Finally, if~$e_\phi$ lies on a cycle, it reduces the overall probability by~$\frac{1}{q-1}$ due to~\eqref{equation:alphaphi2} and by~$\frac{1}{q-2}$ due to~\eqref{equation:alphaphi1}. Therefore, we have the following, where~$u$ denotes the number of edge-node incidences in~$T$, and~$k$ denotes the number of back edges in a BFS run (which is identical in every run of a BFS algorithm).
    \begin{itemize}
        \item If~$e_\phi$ is not on a cycle in~$T$ then $\Pr((Q^T\vert\phi)=A)=\left(\frac{1}{q-1}\right)^{u-k}$.
        \item If~$e_\phi$ is on a cycle in~$T$ then $\Pr((Q^T\vert\phi)=A)=\left(\frac{1}{q-1}\right)^{u-k}\cdot\frac{1}{q-2}$.\qedhere
    \end{itemize}
\end{proof}


\section{Choice of sets}\label{appendix:omitted}

The process of choosing the sets~$\{ J^{(j,i)} \}_{(j,i)\in[r]\times [b]}$ in~\eqref{equation:JsNonSystematic} is very simple, and is best illustrated by the following examples. 
\begin{example}
    Assume that~$N-K=4$ and~$K=6$, which implies that~$r=3$ and~$b=2$. Consider the following matrix
    \begin{align*}
        \begin{pmatrix}
            1 & 1 & 1 & 1 &   & \\ 
            2 & 2 &   &   & 1 & 1 \\
              &   & 2 & 2 & 2 & 2 \\
        \end{pmatrix},
    \end{align*}
    which naturally corresponds to the sets
    \begin{align*}
        J^{(1,1)}& =\{1,2,3,4\} & J^{(1,2)}&=\varnothing\\
        J^{(2,1)}& =\{5,6\} & J^{(2,2)}&=\{1,2\}\\
        J^{(3,1)}& =\varnothing & J^{(3,2)}&=\{3,4,5,6\}.
    \end{align*}
\end{example}
As another example, in which~$N-K\ge K$, we may consider the following.
\begin{example}
    Assume that~$N-K=6$ and~$K=4$, which implies that~$r=2$ and~$b=3$. Consider the following matrix
    \begin{align*}
        \begin{pmatrix}
            1 & 1 & 1 & 1 & 2 & 2 \\
            2 & 2 & 3 & 3 & 3 & 3 
        \end{pmatrix}
    \end{align*}
        which naturally corresponds to the sets
        \begin{align*}
        J^{(1,1)}& =\{1,2,3,4\} & J^{(2,1)}&=\varnothing\\
        J^{(1,2)}& =\{5,6\} & J^{(2,2)}&=\{1,2\}\\
        J^{(1,3)}& =\varnothing & J^{(2,3)}&=\{3,4,5,6\}.
        \end{align*}
\end{example}

\end{document}